\documentclass[runningheads]{llncs}
\usepackage{graphicx}
\usepackage{xcolor}
\usepackage{amsmath}
\newcommand{\ra}{\rightarrow}
\newcommand{\RG}{K_r}
\newtheorem{observation}{Observation}

\newcommand{\p}{{\rm P}}
\newcommand{\np}{{\rm NP}}
\newcommand{\conp}{{\rm coNP}}

\begin{document}
\title{The Complexity of $(P_k, P_\ell)$-Arrowing}
%
%
\author{Zohair Raza Hassan\and
Edith Hemaspaandra \and
Stanis\l aw  Radziszowski}
\authorrunning{Z.R. Hassan, E. Hemaspaandra, S. Radziszowski}
%
\institute{Rochester Institute of Technology, Rochester NY 14623, USA \\ \email{zh5337@rit.edu}, \email{\{eh,spr\}@cs.rit.edu}}
\maketitle              
\begin{abstract}
For fixed nonnegative integers $k$ and $\ell$, the $(P_k, P_\ell)$-Arrowing problem asks whether a given graph, $G$, has a red/blue coloring of $E(G)$ such that there are no red copies of $P_k$ and no blue copies of $P_\ell$. The problem is trivial when $\max(k,\ell) \leq 3$, but has been shown to be coNP-complete when $k = \ell = 4$. In this work, we show that the problem remains coNP-complete for all pairs of $k$ and $\ell$, except $(3,4)$, and when $\max(k,\ell) \leq 3$. 

Our result is only the second hardness result for $(F,H)$-Arrowing for an infinite family of graphs and the first for 1-connected graphs. Previous hardness results for $(F, H)$-Arrowing depended on constructing graphs that avoided the creation of too many copies of $F$ and $H$, allowing easier analysis of the reduction. This is clearly unavoidable with paths and thus requires a more careful approach. We define and prove the existence of special graphs that we refer to as ``transmitters.'' Using transmitters, we construct gadgets for three distinct cases: 1) $k = 3$ and $\ell \geq 5$, 2) $\ell > k \geq 4$, and 3) $\ell = k \geq 4$. For $(P_3, P_4)$-Arrowing we show a polynomial-time algorithm by reducing the problem to 2SAT, thus successfully categorizing the complexity of all $(P_k, P_\ell)$-Arrowing problems. 
\end{abstract}

\section{Introduction and Related Work}

Often regarded as the study of how order emerges from randomness, Ramsey theory has played an important role in mathematics and computer science; it has applications in several diverse fields, including, but not limited to, 
game theory, information theory, and approximation algorithms~\cite{rosta}.
A key operator within the field is the arrowing operator:
given graphs $F, G$, and $H$, we say that $G \ra (F, H)$ (read, $G$ \textit{arrows} $F, H$) if every red/blue edge-coloring of $G$'s edge contains a red $F$ or a blue $H$. 
In this work, we categorize the computational complexity of evaluating this operator when $F$ and $H$ are fixed path graphs. The problem is defined formally as follows.

\begin{problem}[$(F, H)$-Arrowing]
Let $F$ and $H$ be fixed graphs.
Given a graph $G$, does $G \ra (F, H)$?
\end{problem}

The problem is clearly in coNP; a red/blue coloring of $G$ with no red $F$'s and no blue $H$'s forms a certificate that can be verified in polynomial time since $F$ and $H$ are fixed graphs.
We refer to such a coloring as an $(F, H)$-good coloring.
The computational complexity of $(F, H)$-Arrowing has been categorized for a number of pairs $(F, H)$, with a significant amount of work done in the 80s and 90s.
Most relevant to our work is a result by Rutenburg, who showed that $(P_4, P_4)$-Arrowing is coNP-complete~\cite{rut:c:graph-coloring}, where $P_n$ is the path graph on $n$ vertices.
Burr showed that $(F, H)$-Arrowing is in P when $F$ and $H$ are star graphs or when $F$ is a matching~\cite{Bu3}.
Using ``senders''---graphs with restricted $(F, H)$-good colorings introduced a few years earlier by Burr et al.~\cite{burr1976graphs,burr1985use}, Burr showed that $(F, H)$-Arrowing is coNP-complete when $F$ and $H$ are members of $\Gamma_3$, the family of all 3-connected graphs and $K_3$. 
The generalized $(F,H)$-Arrowing problem, where $F$ and $H$ are also part of the input, was shown to be $\Pi^p_2$-complete by Schaefer~\cite{Scha}.\footnote{$\Pi_2^p = \conp^{\np}$,
the class of all problems whose complements are solvable
by a nondeterministic polynomial-time Turing machine having
access to an NP oracle~\cite{mey-sto-ph}.}
Aside from categorizing complexity, the primary research avenue concerned with the arrowing operator is focused on finding minimal---with different possible definitions of minimal---graphs for which arrowing holds. The smallest orders of such graphs are referred to as Ramsey numbers.
Folkman numbers are defined similarly for graphs with some extra structural constraints.
We refer the interested reader to surveys by Radziszowski~\cite{ds1}
and Bikov~\cite{bikov2018} for more information on Ramsey numbers and Folkman numbers, respectively.

Our work provides the first complexity result for $(F, H)$-Arrowing for an infinite family of graphs since Burr's $\Gamma_3$ result from 1990.
It is important to note that Burr's construction relies on that fact that contracting less than three vertices between pairs of 3-connected graphs does not create new copies of said graph. Let $F$ be 3-connected and $u,v \in V(F)$.
Construct $G$ by taking two copies of $F$ and contracting $u$ across both copies, then contracting $v$ across both copies.
Observe 
that no new copies of $F$ are constructed in this process; if a new $F$ is created then it must be disconnected by the removal of the two contracted vertices, contradicting $F$'s 3-connectivity.
This process does not work for paths since contracting two path graphs will always make several new paths across the vertices of both paths.
Thus, we require a more careful approach when constructing the gadgets necessary for our reductions.
We focus on the problem defined below and
prove a dichotomy theorem
categorizing the problem to be in \p~or be \conp-complete. We note that such theorems for other graph problems exist in the 
literature, e.g.,~\cite{ach:j:colorability,for-hop-wyl:j:subgraph-homeomorphism,hel-nes:j:H-coloring,le_et_al:LIPIcs.MFCS.2022.68}.
\begin{problem}[$(P_k, P_\ell)$-Arrowing]
Let $k$ and $\ell$ be fixed integers such that
$2 \leq k \leq \ell$.
Given a graph $G$, does $G \ra (P_k, P_\ell)$?
\end{problem}

\begin{theorem}
\label{thm:main}
    $(P_k, P_\ell)$-Arrowing is \conp-complete for all $k$ and $\ell$ unless $k = 2$, $(k, \ell) = (3,3)$, or $(k, \ell) = (3, 4)$. For these exceptions, the problem is in \p.
\end{theorem}

Before this, the only known coNP-complete case for paths was when $k = \ell = 4$~\cite{rut:c:graph-coloring}. 
Despite being intuitively likely, generalizing the hardness result to larger paths proved to be an arduous task. 
Our proof relies on proving the existence of graphs with special colorings---we rely heavily on work by Hook~\cite{Ho}, who categorized the $(P_k, P_\ell)$-good colorings of the largest complete graphs which do not arrow $(P_k, P_\ell)$. 
After showing the existence of these graphs, the reduction is straightforward. 
The polynomial-time cases are straightforward (Theorem~\ref{t:p-easy}) apart from the case where $(k, \ell) = (3,4)$, wherein we reduce the problem to 2SAT (Theorem~\ref{thm:p3p4}).

The rest of this paper is organized as follows. We present the necessary preliminaries in Section~\ref{sec:prelim}. 
The proof for Theorem~\ref{thm:main} is split into Sections~\ref{sec:poly} (the polynomial-time cases) and~\ref{sec:conp} (the \conp-complete cases). 
We conclude in Section~\ref{sec:conclude}.

\section{Preliminaries}
\label{sec:prelim}

All graphs discussed in this work are simple and undirected. $V(G)$ and $E(G)$ denote the vertex and edge set of a graph $G$, respectively. We denote an edge in $E(G)$ between $u,v \in V(G)$ as $(u,v)$.
For two disjoint subsets $A, B \subset V(G)$, $E(A,B)$ refers to the edges with one vertex in $A$ and one vertex in $B$.
The neighborhood of a vertex $v \in V(G)$ is denoted as $N(v)$ and $d(v) := |N(v)|$.
The path, cycle, and complete graphs on $n$ vertices are denoted as $P_n$, $C_n$, and $K_n$, respectively. The complete graph on $n$ vertices missing an edge is denoted as $K_n - e$.
Vertex contraction is the process of replacing two vertices $u$ and $v$ with a new vertex $w$ such that $w$ is adjacent to all remaining neighbors $N(u) \cup N(v)$.

An $(F,H)$-good coloring of a graph $G$ is a red/blue coloring of $E(G)$ where the red subgraph is $F$-free, and the blue subgraph is $H$-free. We say that $G$ is $(F, H)$-good if it has at least one $(F,H)$-good coloring.
When the context is clear, we will omit $(F, H)$ and refer to the coloring as a good coloring. 

Formally, a coloring for $G$ is defined as function $c : E(G) \rightarrow \{\text{red}, \text{blue}\}$ that maps edges to the colors red and blue. For an edge $(u,v)$ and coloring $c$, we denote its color as $c(u,v)$.

\section{Polynomial-Time Cases}
\label{sec:poly}

In this section, we prove the P cases from Theorem~\ref{thm:main}.
Particularly, we describe polynomial-time algorithms for $(P_2, P_\ell)$-Arrowing and $(P_3, P_3)$-Arrowing (Theorem~\ref{t:p-easy}) and provide a polynomial-time reduction from $(P_3, P_4)$-Arrowing to 2SAT (Theorem~\ref{thm:p3p4}). 

\begin{theorem}
\label{t:p-easy}
$(P_k, P_\ell)$-Arrowing is in $\p$ when $k = 2$ and when $k = \ell = 3$.
\end{theorem}

\begin{proof}
Let $G$ be the input graph. Without loss of generality, assume that $G$ is connected (for disconnected graphs, we run the algorithm on each connected component).

\noindent \textit{Case 1 $(k = 2)$.}  
Coloring any edge in $G$ red will form a red $P_2$. Thereby, the entire graph must be colored blue. Thus, a blue $P_\ell$ is avoided if and only if $G$ is $P_\ell$-free, which can be checked by brute force, since $\ell$ is constant.

\noindent \textit{Case 2 $(k = \ell = 3)$.}  
Note that in any $(P_3, P_3)$-good coloring of $G$, edges of the same color 
    cannot be adjacent; otherwise, a red or blue $P_3$ is formed. 
    Thus, we can check if $G$ is $(P_3, P_3)$-good similarly to how we check if a graph is 2-colorable: arbitrarily color an edge red and color all of its adjacent edges blue. For each blue edge, color its neighbors red and for each red edge, color its neighbors blue. 
    Repeat this process until all edges are colored or a red or blue $P_3$ is formed.
    This algorithm is clearly polynomial-time. \qed
\end{proof}

The proof that $(P_3, P_4)$-Arrowing is in P consists of two parts. A preprocessing step to simplify the graph (using Lemmas~\ref{lem:p3p4-1} and~\ref{lem:p3p4-2}), followed by a reduction to 2SAT, which was proven to be in P by Krom in 1967~\cite{krom1967decision}. 

\begin{problem}[2SAT]
    Let $\phi$ be a CNF formula where each clause has at most two literals. Does there exist a satisfying assignment of $\phi$?
\end{problem}

\begin{lemma}
\label{lem:p3p4-1}
Suppose $G$ is a graph and $v \in V(G)$ is a vertex such that
$d(v) = 1$ and $v$'s only neighbor has degree at most two.
Then,
$G$ is $(P_3, P_4)$-good if and only if $G - v$ is $(P_3, P_4)$-good.
\end{lemma}
\begin{proof}
Let $u$ be the neighbor of $v$.
If $d(u) = 1$, the connected component of $v$ is a $K_2$ and the statement is trivially true.
If $d(u) = 2$, let $w$ be the other neighbor of $u$, i.e., the neighbor that is not $v$.
Clearly, if 
$G$ is $(P_3, P_4)$-good, then $G - v$ is $(P_3, P_4)$-good. 
We now prove the other direction. Suppose we have good coloring of $G - v$. It is immediate that we can extend this to a good coloring of $G$ by coloring $(v,u)$ (the only edge adjacent to $v$) red if $(u,w)$ is colored blue, and blue if $(u,w)$ is colored red.
 \qed
\end{proof}

\begin{lemma}
\label{lem:p3p4-2}
Suppose $G$ is a graph and there is a $P_4$ in $G$ with edges $(v_1, v_2), (v_2, v_3)$, and $(v_3, v_4)$ such that $d(v_1) = d(v_2) = d(v_3) = d(v_4) = 2$. 
Then,
$G$ is $(P_3, P_4)$-good if and only if $G - v_2$ is $(P_3, P_4)$-good.
\end{lemma}

\begin{proof}
If $(v_1,v_4)$ is an edge, then the connected component of $v_2$ is a $C_4$ and the statement is trivially true.
If not, let $v_0, v_5 \not \in \{v_1, v_2, v_3, v_4\}$ be such that $(v_0,v_1)$ and 
$(v_4,v_5)$ are edges. Note that it is possible that $v_0 = v_5$. 
Clearly, if
$G$ is $(P_3, P_4)$-good then $G - v_2$ is $(P_3, P_4)$-good. For the other direction, suppose $c$ is a $(P_3, P_4)$-good coloring of $G - v_2$. We now construct a coloring $c'$ of $G$. We  color all edges other than $(v_1, v_2), (v_2, v_3)$, and $(v_3, v_4)$
the same as $c$. 
The colors of the remaining three edges are determined by the coloring of $(v_0,v_1)$ and $(v_4,v_5)$ as follows.

\begin{itemize}
\item If $c(v_0,v_1) = c(v_4,v_5) =$ red, then
$c'(v_1,v_2), c'(v_2,v_3), c'(v_3,v_4) = $ blue, red, blue.
\item If $c(v_0,v_1) = c(v_4,v_5) =$ blue, then
$c'(v_1,v_2), c'(v_2,v_3), c'(v_3,v_4) = $ red, blue, red.
\item If $c(v_0,v_1) = $ red and $c(v_4,v_5) =$ blue, then
$c'(v_1,v_2), c'(v_2,v_3), c'(v_3,v_4) = $ blue, blue, red.
\item If $c(v_0,v_1) = $ blue and $c(v_4,v_5) =$ red, then
$c'(v_1,v_2), c'(v_2,v_3), c'(v_3,v_4) = $ red, blue, blue. 
\end{itemize}

Since the cases above are mutually exhaustive, this completes the proof. \qed
\end{proof}

\begin{theorem}
\label{thm:p3p4}
    $(P_3, P_4)$-Arrowing is in $\p$.
\end{theorem}
\begin{proof}
Let $G$ be the input graph. 
Let $G'$ be the graph obtained by repeatedly removing vertices $v$ described in Lemma~\ref{lem:p3p4-1} and
vertices $v_2$ described in Lemma~\ref{lem:p3p4-2} until no more such vertices exist. As implied by said lemmas, $G' \ra (P_3, P_4)$ if and only if $G \ra (P_3, P_4)$. Thus, it suffices to construct a 2SAT formula $\phi$ such that $\phi$ is satisfiable if and only if $G'$ is $(P_3, P_4)$-good.

Let $r_e$ be a variable corresponding to the edge $e \in E(G')$, denoting that $e$ is colored red.
We construct a formula $\phi$, where a solution to $\phi$ corresponds to a coloring of $G'$.
For each $P_3$ in $G'$, with edges $(v_1, v_2)$ and $(v_2, v_3)$, add the clause 
$\left( \overline{r_{(v_1, v_2)}} \lor \overline{r_{(v_2, v_3)}} \right)$. Note that this expresses ``no red $P_3$'s.'' 
For each $P_4$ in $G'$, with edges $(v_1, v_2), (v_2, v_3)$, and $(v_3, v_4)$:

\begin{enumerate}
    \item If $( v_2, v_4 ) \in E(G')$, add the clause
    $\left( r_{( v_1, v_2 )} \lor r_{( v_3, v_4 )} \right) $.
    
    \item If $( v_2, v_4 ) \not \in E(G')$ and $d(v_2) > 2$, then add the clause
    $\left ( r_{( v_2, v_3 )} \lor r_{(v_3, v_4 )} \right ) $.
\end{enumerate}

It is easy to see that the conditions specified above must be satisfied by each good coloring of $G'$, and thus 
$G'$ being $(P_3, P_4)$-good implies that $\phi$ is satisfiable.
We now prove the other direction by contradiction.
Suppose $\phi$ is satisfied, but the corresponding coloring $c$ is not $(P_3, P_4)$-good. It is immediate that red $P_3$'s cannot occur in $c$, so we assume that there exists a blue $P_4$, with edges
$e = (v_1, v_2), f = (v_2, v_3)$, and $g = (v_3, v_4)$ such that $r_e = r_f = r_g = $ false in the satisfying assignment of $\phi$. Without loss of generality, assume that $d(v_2) \geq d(v_3)$. 

\begin{itemize}
\item If $d(v_2) > 2$, $\phi$ would contain clause $r_e \vee r_g$ or $r_f \vee r_g$. It follows that $d(v_2) = d(v_3) = 2$.
\item If $d(v_1) = 1$, $v_1$ would have been deleted by applying Lemma~\ref{lem:p3p4-1}. It follows that $d(v_1) > 1$. Similarly, $d(v_4) > 1$.
\item If $d(v_1) > 2$, then there exists a vertex $v_0$ such that $(v_0, v_1), (v_1, v_2), (v_2,v_3)$ are a $P_4$ in $G'$, $d(v_1) > 2$ and $(v_1, v_3) \not \in E(G')$ (since $d(v_3) = 2)$. This implies that $\phi$ contains clause $r_e \vee r_f$, which is a contradiction. It follows that $d(v_1) = 2$. Similarly, $d(v_4) = 2$.
\item So, we are in the situation that $d(v_1) = d(v_2) = d(v_3) = d(v_4) = 2$. But then $v_2$ would have been deleted by Lemma~\ref{lem:p3p4-2}.
\end{itemize}
Since the cases above are mutually exhaustive, this completes the proof. \qed
\end{proof}

\section{coNP-Complete Cases}
\label{sec:conp}

In this section, we discuss the coNP-complete cases in Theorem~\ref{thm:main}. 
In Section~\ref{sec:reduc}, we describe how NP-complete SAT variants can be reduced to $(P_k, P_\ell)$-Nonarrowing (the complement of $(P_k, P_\ell)$-Arrowing: does there exist a $(P_k, P_\ell)$-good coloring of $G$?).
The \np-complete SAT variants are defined below.

\begin{problem}[$(2,2)$-3SAT~\cite{berman200322sat}]
Let $\phi$ be a CNF formula where each clause contains
exactly three distinct variables,
and each variable appears only
four times: twice unnegated and twice negated. Does there exist a satisfying assignment for $\phi$?
\end{problem}

\begin{problem}[Positive NAE E3SAT-4~\cite{antunes2019characterizing}]
    Let $\phi$ be a CNF formula where each clause is an NAE-clause (a clause that is satisfied when its literals are not all true or all false) containing exactly three (not necessarily distinct) variables, and each variable appears 
    at most four times, only unnegated. 
    Does there exist a satisfying assignment for $\phi$?
\end{problem}

\begin{figure}[t]
    \centering
    \includegraphics[width=0.8\textwidth]{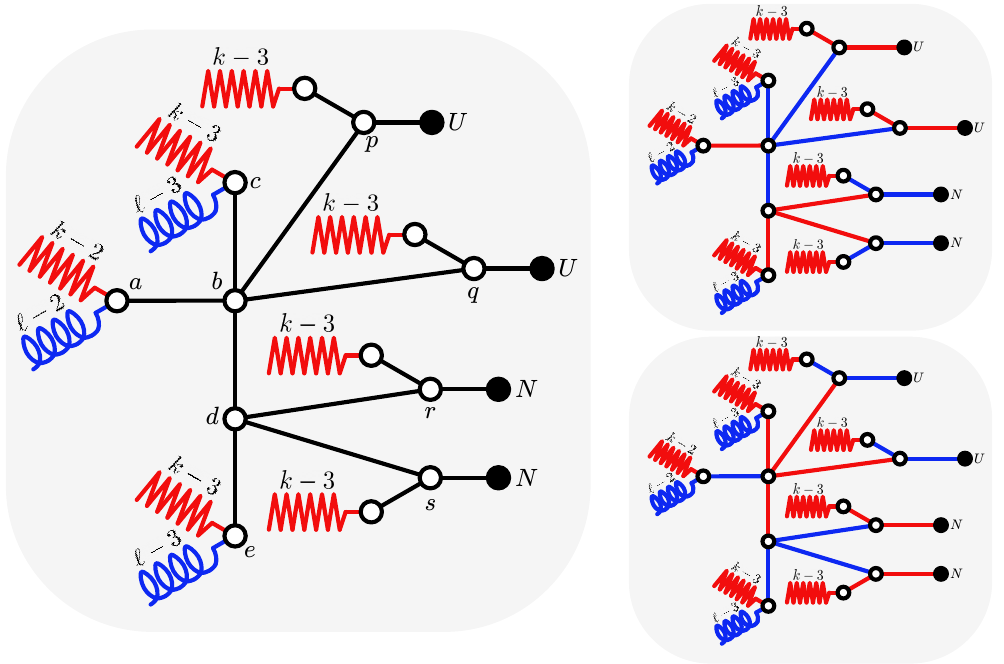}
    \caption{The variable gadget for $(P_k, P_\ell)$-Nonarrowing  when $4 \leq k < \ell$ is shown on the left. The output vertices are filled in. Red jagged lines and blue spring lines represent $(k,\ell,x)$-red- and $(k,\ell,x)$-blue-transmitters, respectively, where the value of $x$ is shown on the top, and the vertex the lines are connected to are the strict endpoints of the monochromatic paths.
    Observe that when $(a, b)$ is red, other edges adjacent to $b$ must be blue to avoid a red $P_{k}$. This, in turn, causes neighbors $p$ and $q$ to have incoming blue $P_{\ell-1}$'s, and vertices marked \textbf{U} are now strict endpoints of red $P_{k-1}$'s.
    Moreover, edges adjacent to $d$ (except $(b,d)$) must be red to avoid blue $P_{\ell}$'s. Thus, $r$ and $s$ are strict endpoints of red $P_{k-1}$'s, causing the vertices marked \textbf{N} to be strict endpoints of blue $P_3$'s.
    A similar pattern is observed when $(a,b)$ is blue.
    Note that for $k \leq 4$, the $(k,\ell,k-3)$-red-transmitter can be ignored.
    On the right, the two kinds of $(P_k, P_\ell)$-good colorings of the gadget are shown.}
    \label{fig:pkpl-var}
\end{figure}

Our proofs depend on the existence of graphs we refer to as ``transmitters,'' defined below.
These graphs 
enforce behavior on
special vertices which are 
\textit{strict endpoints} of red or blue paths.
For a graph $G$ and coloring $c$, we say that $v$ is a strict endpoint of a red (resp., blue) $P_k$ in $c$ if $k$ is the length of the longest red (resp., blue)
path that $v$ is the endpoint of.
We prove the existence of these graphs in Section~\ref{sec:transmit}.

\begin{definition}
Let $3 \leq k < \ell$.
For an integer $x \in  \{2,3,\ldots,k-1\}$
(resp., $x \in \{2,3,\ldots,\ell-1\}$) 
a $(k,\ell,x)$-red-transmitter (resp., $(k,\ell,x)$-blue-transmitter) is a $(P_k, P_\ell)$-good graph $G$ with a vertex $v \in V(G)$ such that in every $(P_k, P_\ell)$-good coloring of $G$, $v$ is the strict endpoint of a 
red (resp., blue) $P_x$, and is not adjacent to any blue (resp., red) edge.
\end{definition}

\begin{definition}
Let $k \geq 3$ and $x \in  \{2,3,\ldots,k-1\}$. 
A $(k,x)$-transmitter is a $(P_k, P_k)$-good graph $G$ with a vertex $v \in V(G)$ such that in every $(P_k, P_k)$-good coloring of $G$, $v$ is either
(1) the strict endpoint of a red $P_x$ and not adjacent to any blue edge, or
(2) the strict endpoint of a blue $P_x$ and not adjacent to any red edge.
\end{definition}

\subsection{Reductions}
\label{sec:reduc}

We present three theorems that describe 
gadgets to reduce NP-complete variants of SAT to $(P_k, P_\ell)$-Nonarrowing.

\begin{figure}[t]
    \centering
    \includegraphics[width=0.8\textwidth]{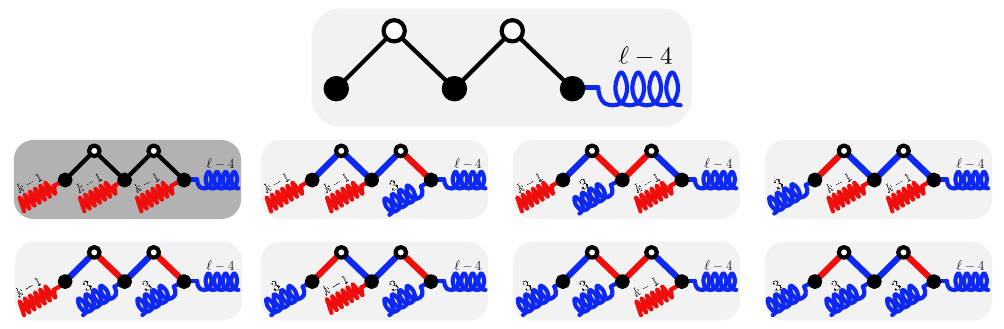}
    \caption{The clause gadget for $(P_k, P_\ell)$-Nonarrowing when $4 \leq k < \ell$ is shown on top. The input vertices are filled in.
    Below it, we show the eight possible combinations of inputs that can be given to the gadget. Observe that a $(P_k, P_\ell)$-good coloring is always possible unless the input is three red $P_{k-1}$'s (top left). 
    As in Figure~\ref{fig:pkpl-var}, jagged and spring lines represent transmitters.
    We use this representation of transmitters to depict the two forms of input to the gadget.
    For $\ell \leq 5$, the $(k,\ell,\ell-4)$-blue-transmitter can be ignored.}
    \label{fig:pkpl-clause}
\end{figure}

\begin{theorem}
\label{thm:pkpl}
    $(P_k, P_\ell)$-Arrowing is \conp-complete for all $4 \leq k < \ell $.
\end{theorem}

\begin{proof}
    We reduce $(2, 2)$-3SAT to $(P_k, P_\ell)$-Nonarrowing.
    Let $\phi$ be the input to $(2, 2)$-3SAT. We construct $G_\phi$ such that $G_\phi$ is $(P_k, P_\ell)$-good if and only if $\phi$ is satisfiable.
    Let $VG$ and $CG$ be the variable and clause gadgets shown in Figures~\ref{fig:pkpl-var} and~\ref{fig:pkpl-clause}.
    $VG$ has four output vertices that emulate the role of 
    sending a truth signal from a variable to a clause. 
    We first look at Figure~\ref{fig:pkpl-var}.
    The vertices labeled $U$ (resp., $N$) correspond to unnegated (resp., negated) signals.
    Being the strict endpoint of a blue $P_3$ corresponds to a true signal while being the strict endpoint of a red $P_{k-1}$ corresponds to a false signal.
    We now look at Figure~\ref{fig:pkpl-clause}.
    When three red $P_{k-1}$ signals are sent to the clause gadget, it forces the entire graph to be blue, forming a blue $P_\ell$. When at least one blue $P_3$ is present, a good coloring of $CG$ is possible. 
    
    We construct $G_\phi$ like so.
    For each variable (resp., clause) in $\phi$, we add 
    a copy of $VG$ (resp., $CG$) to $G_\phi$.
    If a variable appears unnegated (resp., negated) in a clause, a $U$-vertex (resp., $N$-vertex) from the corresponding $VG$ is contracted with a previously uncontracted input vertex of the $CG$ corresponding to said clause.
    The correspondence between satisfying assignments of $\phi$ and good colorings of $G_\phi$ is easy to see. 
    \qed
\end{proof}

\begin{figure}[t]
    \centering
    \includegraphics[width=0.8\textwidth]{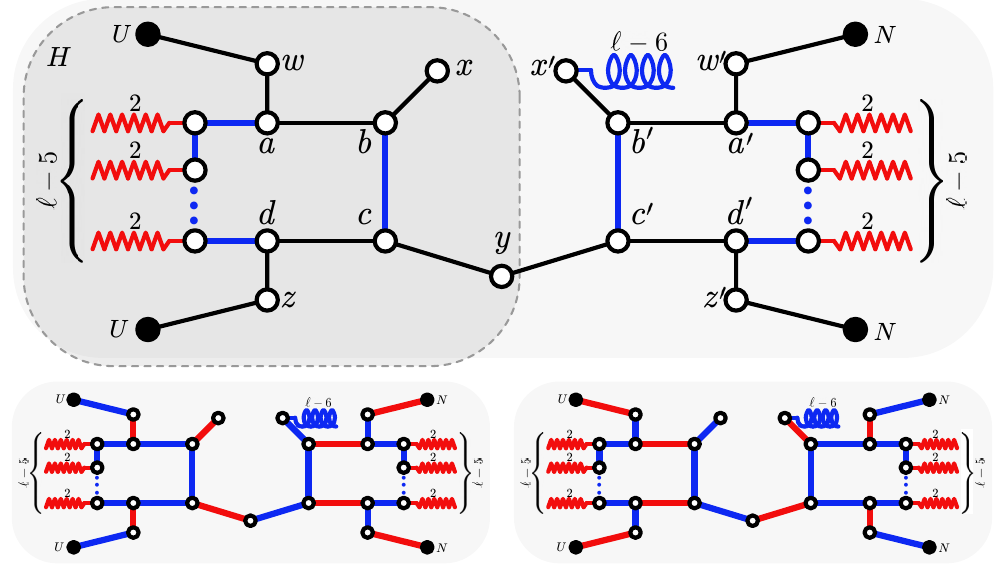}
    \caption{The variable gadget for $(P_3, P_\ell)$-Nonarrowing where $\ell \geq 6$ (top) and its two good colorings (bottom).
    The variable gadget is a combination of two $H$'s, whose properties we discuss in the proof of Theorem~\ref{thm:p3pl}. Note that when $(a,b)$ is red in $H$, then $(a',b')$ is blue in $H$'s copy, and vice versa; if both copies have the same coloring of $(a,b)$, then a red $P_3$ is formed at $y$, 
    or a blue $P_\ell$ is formed from the path from $x$ to $x'$ and the $(3,\ell,\ell-6)$-blue-transmitter that $x'$ is connected to.
    When $\ell = 5$, the edge $(a,d)$ is added in $H$, in lieu of the $\ell-5$ vertices connected to $(3,\ell,2)$-red-transmitters. 
    Note that for $\ell \leq 8$, the $(3,\ell,\ell-6)$-blue-transmitter can be ignored.
    }
    \label{fig:p3pl}
\end{figure}

\begin{theorem}
\label{thm:p3pl}
    $(P_3, P_\ell)$-Arrowing is \conp-complete for all $\ell \geq 5$.
\end{theorem}

\begin{proof}
We proceed as in the proof of Theorem~\ref{thm:pkpl}. The variable gadget is shown in Figure~\ref{fig:p3pl}. Blue (resp., red) $P_2$'s incident to vertices marked \textbf{U} and \textbf{N} correspond to true (resp., false) signals. The clause gadget is the same as Theorem~\ref{thm:pkpl}'s, but the good colorings are different since the inputs are red/blue $P_2$'s instead. These colorings are shown in the appendix in Figure~\ref{fig:p3pl-clause}.

Suppose $\ell \geq 6$.
Let $H$ be the graph circled with a dotted line in Figure~\ref{fig:p3pl}. We first discuss the properties of $H$.
Note that any edge adjacent to a red $P_2$ must be colored blue to avoid a red $P_3$.
Let $v_1, v_2, \ldots, v_{\ell-5}$ be the vertices connected to $(3,\ell,2)$-red-transmitters such that $v_1$ is adjacent to $a$.
Observe that $(a,b)$ and $(c,d)$ must always be the same color; if, without loss of generality, $(a,b)$ is red and $(c,d)$ is blue, a blue $P_\ell$ is formed via the sequence $a, v_1, \ldots, v_{\ell-5}, d, c, b, x$.
In the coloring where $(a,b)$ and $(c,d)$ are blue, the vertices $a, v_1, \ldots, v_{\ell-5}, d, c, b$ form a blue $C_{\ell-1}$, and all edges going out from the cycle must be colored red to avoid blue $P_{\ell}$'s. 
This forces the vertices marked \textbf{U} to be strict endpoints of blue $P_2$'s.
If $(a,b)$ and $(c,d)$ are red, $w,a,v_1,\ldots,v_{\ell-5},d,z$ forms a blue $P_{\ell-1}$, forcing the vertices marked \textbf{U} to be strict endpoints of red $P_2$'s. Moreover, $(x,b)$ and $(y,c)$ must also be blue.

With these properties of $H$ in mind, the functionality of the variable gadget described in Figure~\ref{fig:p3pl}'s caption is easy to follow. The $\ell = 5$ case uses a slightly different $H$, also described in the caption.
    \qed
\end{proof}

\begin{theorem}
\label{thm:pkpk}
    $(P_k, P_k)$-Arrowing is \conp-complete for all $k \geq 4$.
\end{theorem}
\begin{figure}[t]
    \centering
    \includegraphics[width=0.8\textwidth]{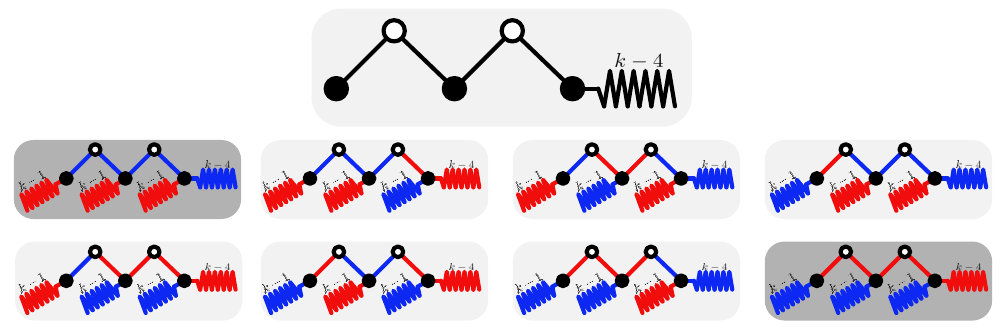}
    \caption{The clause gadget for $(P_k, P_k)$-Nonarrowing. The format is similar to Figure~\ref{fig:pkpl-clause}.}
    \label{fig:pkpk-clause}
\end{figure}

\begin{figure}[t]
    \centering
    \includegraphics[width=0.8\textwidth]{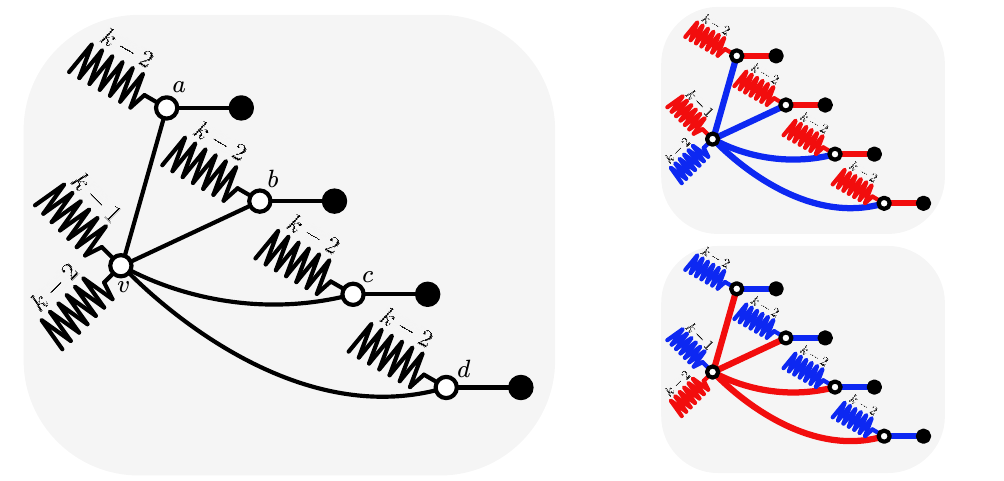}
    \caption{The variable gadget for $(P_k, P_k)$-Nonarrowing. Observe that the transmitters connected to $v$ must have different colors; otherwise, a red or blue $P_{k-1 + k-2 -1}$ is formed, which is forbidden when $k \geq 4$. When the $(k,k-1)$-transmitter is red, $v$'s other neighboring edges must be blue. Thus, vertices $a,b,c$, and $d$ are strict endpoints of blue $P_{k-1}$'s, causing the output vertices (filled) to be strict endpoints of red $P_{k-1}$'s. A similar situation occurs when the $(k,k-1)$-transmitter is blue. Both $(P_k, P_k)$-good colorings are shown on the right.}
    \label{fig:pkpk}
\end{figure}
\begin{proof}
    For $k=4$, Rutenburg showed that the problem is coNP-complete by providing gadgets that reduce from an NAE SAT variant~\cite{rut:c:graph-coloring}. For $k \geq 5$, we take a similar approach and reduce
    Positive NAE E3SAT-4
    to $(P_k, P_k)$-Nonarrowing using the clause and variable gadgets described in Figures~\ref{fig:pkpk-clause} and~\ref{fig:pkpk}.
    The variable gadget has four output vertices, all of which are unnegated. Without loss of generality, we assume that blue $P_{k-1}$'s correspond to 
    true signals.
    The graph $G_\phi$ is constructed as in the proofs of Theorems~\ref{thm:pkpl} and~\ref{thm:p3pl}.
    Our variable gadget is still valid when $k = 4$, but the clause gadget does not admit a $(P_4,P_4)$-good coloring for all the required inputs. In Figure~\ref{fig:p4p4} in the appendix, we show a different clause gadget that can be used to show the hardness of $(P_4, P_4)$-Arrowing using our reduction.
    \qed
\end{proof}

\subsection{Existence of Transmitters}
\label{sec:transmit}

Our proofs for the existence of transmitters are corollaries of the following.

\begin{lemma}
\label{lem:k-l-transmit}
For integers $k, \ell$,
where $3 \leq k < \ell$,
$(k,\ell,k-1)$-red-transmitters exist.
\end{lemma}

\begin{lemma}
\label{lem:k-k-transmit}
For all $k \geq 3$,
$(k,k-1)$-transmitters exist.
\end{lemma}

In the interest of saving space, we only present the proof of one case (when $k$ is even) of Lemma~\ref{lem:k-l-transmit} in our main text, and defer the rest to the appendix. 
We construct these transmitters by
carefully combining copies of complete graphs.
The Ramsey number $R(P_k, P_\ell)$ is defined as the smallest number $n$ such that $K_n \ra (P_k, P_\ell)$.
We know that $R(P_k, P_\ell) = \ell + \lfloor k/2 \rfloor - 1$, where 
$2 \leq k \leq \ell$~\cite{gerencser1967ramsey}. 
In 2015,
Hook characterized the $(P_k, P_\ell)$-good 
colorings of all ``critical'' complete graphs: 
$K_{R(P_k, P_\ell) - 1}$.
We summarize Hook's results below.\footnote{We note that Hook's ordering convention differs from ours, i.e., they look at $(P_\ell, P_k)$-good colorings. Moreover, they use $m$ and $n$ in lieu of $k$ and $\ell$.}

\begin{figure}[t]
    \centering
    \includegraphics[width=0.8\textwidth]{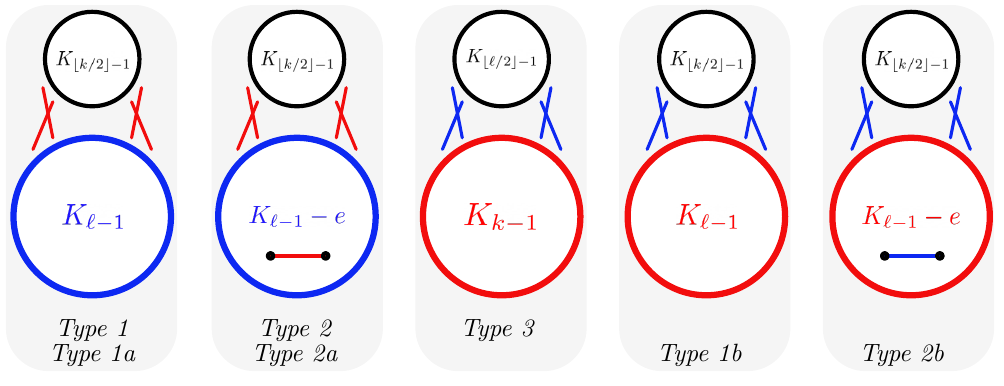}
    \caption{Illustrations of $(P_k, P_\ell)$-good colorings of $K_{R(P_k, P_\ell) - 1}$.}
    \label{fig:kr-col}
\end{figure}

\begin{theorem}[Hook~\cite{Ho}]
\label{thm:CritK-uneq}
Let $4 \leq k < \ell$ and $r = R(P_k, P_\ell) - 1$.
The possible $(P_k, P_\ell)$-good colorings of $\RG$ can be categorized into three types. In each case, $V(G)$ is partitioned into sets $A$ and $B$. The types are defined as follows:
\begin{itemize}
    \item Type 1. 
    Let $|A| = \lfloor k/2 \rfloor - 1$ and $|B| = \ell - 1$.
    Each edge in $E(B)$ must be blue, and each edge in $E(A, B)$ must be red. Any coloring of $E(A)$ is allowed.
    \item Type 2. 
    Let $|A| = \lfloor k/2 \rfloor - 1$ and $|B| = \ell - 1$, and let $b \in E(B)$. Each edge in $E(B) \setminus \{b\}$ must be blue, and each edge in $E(A, B) \cup \{b\}$ must be red. Any coloring of $E(A)$ is allowed.
    \item Type 3. Let $|A| = \lfloor \ell/2 \rfloor - 1$ and $|B| = k - 1$. Each edge in $E(B)$ must be blue, and each edge in $E(A, B)$ must be red. Any coloring of $E(A)$ is allowed.
\end{itemize}
Moreover, the types of colorings allowed vary according to the parity of $k$.
If $k$ is even, then $\RG$ can only have Type 1 colorings. If $k$ is odd and $\ell > k + 1$, then $\RG$ can only have Type 1 and 2 colorings. If $k$ is odd and $\ell = k + 1$, then $\RG$ can have all types of colorings.
\end{theorem}

For the case where $k = \ell$, $K_r$ can have Type 1 and 2 colorings as described in the theorem above. Due to symmetry, the colors in these can be swapped and are referred to as Type 1a, 1b, 2a, and 2b colorings.
The colorings described have been illustrated in Figure~\ref{fig:kr-col}.
We note the following useful observation.

\begin{observation}
\label{obs:k-l-types}
Suppose $\ell > k \geq 4$ and $r = R(P_k, P_\ell) - 1$. 
\begin{itemize}
    \item In Type 1 $(P_k, P_\ell)$-good colorings of $K_r$: (1) each vertex in $B$ is a strict endpoint of a blue $P_{\ell - 1}$, (2) when $k$ is even (resp., odd), each vertex in $B$ is a strict endpoint of a red $P_{k - 1}$ (resp., $P_{k-2}$), and (3) when $k$ is even (resp., odd), each vertex in $A$ is a strict endpoint of a red $P_{k - 2}$ (resp., $P_{k-3}$).
    \item In Type 2 $(P_k, P_\ell)$-good colorings of $K_r$: (1) each vertex in $B$ is a strict endpoint of a blue $P_{\ell - 1}$, (2) each vertex in $B$ is a strict endpoint of a red $P_{k - 1}$, and (3) each vertex in $A$ is a strict endpoint of a red $P_{k - 2}$.
    \item In Type 3 $(P_k, P_\ell)$-good colorings of $K_r$: (1) each vertex in $B$ is a strict endpoint of a red $P_{k - 1}$, (2) each vertex in $B$ is a strict endpoint of a blue $P_{\ell - 1}$, and (3) each vertex in $A$ is a strict endpoint of a blue $P_{\ell - 2}$. 
\end{itemize}
\end{observation}

We justify these claims in the appendix, wherein we also formally define the colorings $K_r$ when $k = \ell$ and justify a similar observation.
Finally, we define a special graph that we will use throughout our proofs.
\begin{definition}[$(H,u,m)$-thread]
\label{def:thread}
Let $H$ be a graph, $u \in V(H)$, and $m \geq 1$ be an integer.
An $(H,u,m)$-thread $G$, is a graph on $m|V(H)| + 1$ vertices constructed as follows.
Add $m$ copies of $H$ to $G$. Let $U_i \subset V(G)$ be the vertex set of the $i^{th}$ copy of $H$, and $u_i$ be the vertex $u$ in $H$'s $i^{th}$ copy. Connect each $u_i$ to $u_{i+1}$ for each $i \in \{1,2,\ldots, m-1\}$. 
Finally, add a vertex $v$ to $G$ and connect it to $u_{m}$. We refer to $v$ as the thread-end of $G$. This graph is illustrated in Figure~\ref{fig:thread}.
\end{definition}

\begin{figure}[t]
    \centering
    \includegraphics[width=\textwidth]{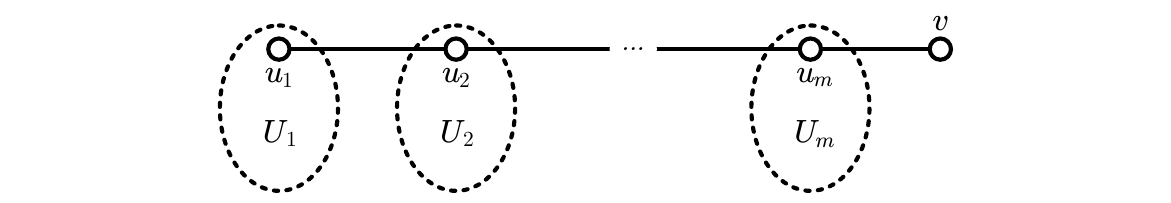}
    \caption{An $(H,u,m)$-thread as described in Definition~\ref{def:thread}.}
    \label{fig:thread}
\end{figure}

Using Theorem~\ref{thm:CritK-uneq}, Observation~\ref{obs:k-l-types}, and Definition~\ref{def:thread}
we are ready to prove the existence of $(k,\ell,k-1)$-red-transmitters and $(k,k-1)$-transmitters via construction.
Transmitters for various cases are shown in Figures~\ref{fig:kl-transmit} and~\ref{fig:kk-transmit}.
We present the proof for one case below and the remaining in the appendix.
\vskip 0.2cm

\noindent\textit{Proof of Lemma~\ref{lem:k-l-transmit} when $k$ is even.}
Let $k \geq 4$ be an even integer and $r = R(P_k, P_\ell) - 1$.
In this case, by Theorem~\ref{thm:CritK-uneq}, only Type 1 colorings are allowed for $\RG$. 
The term $A1$-vertex (resp., $B1$-vertex) is used to refer to vertices belonging to set $A$ (resp., $B$) in a $\RG$ with a Type 1 coloring, as defined in Theorem~\ref{thm:CritK-uneq}.
We first make an observation about the graph $H$, constructed by adding an edge $(u, v)$ between two disjoint $\RG$'s.
Note that $u$ must be an $A1$-vertex, otherwise the edge $(u,v)$ would form a red $P_{k-1}$ or blue $P_{\ell - 1}$ when colored red or blue, respectively (Observation~\ref{obs:k-l-types}).
Similarly, $v$ must also be an $A1$-vertex. Note that $(u,v)$ must be blue; otherwise, by Observation~\ref{obs:k-l-types}, a red $P_{k-2 + k-2}$ is formed, which cannot exist in a good coloring when $k \geq 4$.

\begin{figure}[t]
    \centering
    \includegraphics[width=0.8\textwidth]{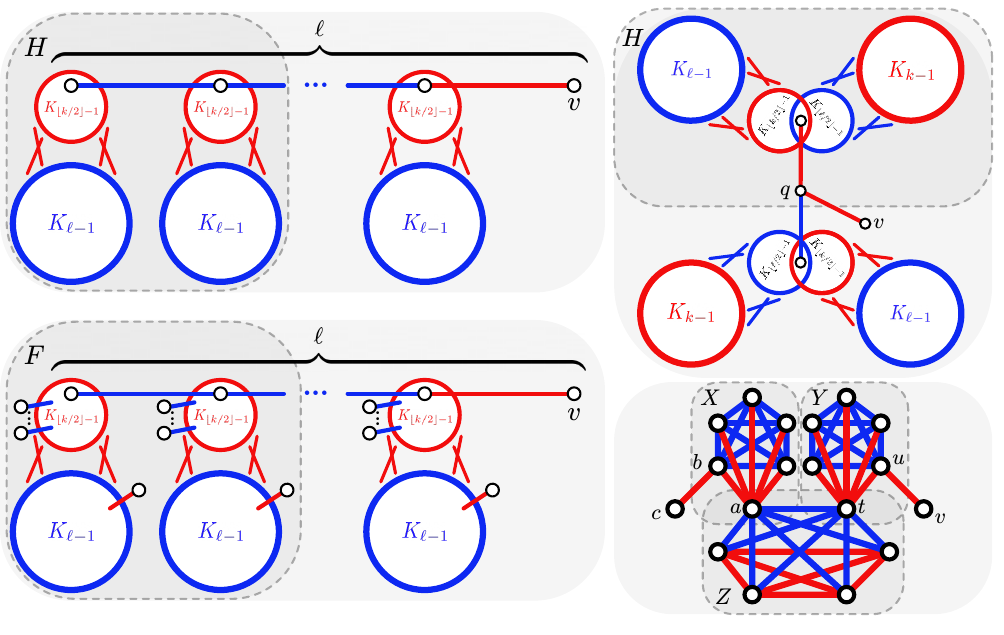}
    \caption{$(k,\ell,k-1)$-red-transmitters for 
    even $k$ with $\ell > k$, odd $k$ with $\ell > k+1$, and odd $k$ with $\ell = k+1$ are shown on the top-left, bottom-left, and top-right, respectively. 
    The latter construction does not work for the case where $k = 5$, so an alternative construction for a $(5,6,4)$-red-transmitter is shown on the bottom-right.
    The graphs ($H$ and $F$)  described in each case are circled so that the proofs are easier to follow. A good coloring is shown for each transmitter.
    }
    \label{fig:kl-transmit}
\end{figure} 

We define the $(k,\ell,k-1)$-red-transmitter, $G$,
as the $(\RG, u, \ell-1)$-thread graph, where $u$ is an arbitrary vertex in $V(\RG)$.
The thread-end $v$ of $G$ is a strict endpoint of a red $P_{k-1}$.
Let $U_i$ and $u_i$ 
be the sets and vertices of $G$ as described in Definition~\ref{def:thread}.
From our observation about $H$, we know that
each edge $(u_i, u_{i+1})$ must be blue. Thus, $u_{\ell-1}$ must be the strict endpoint of a blue $P_{\ell-1}$, implying that $(u_{\ell-1}, v)$ must be red. Since $u_{\ell-1}$ is also a strict endpoint of a red $P_{k-2}$ (Observation~\ref{obs:k-l-types}), $v$ must be the strict endpoint of a red $P_{k-1}$.

For completeness, we must also show that $G$ is $(P_k, P_\ell)$-good.
Let $A_i$ and $B_i$ be the sets $A$ and $B$ as defined in Theorem~\ref{thm:CritK-uneq} for each $U_i$.
Note that the only edges whose coloring we have not discussed are the edges in each $E(A_i)$.
It is easy to see that if each edge in each $E(A_i)$ is colored red, the resulting coloring is $(P_k, P_\ell)$-good. This is because introducing a red edge in $E(A_i)$ cannot form a longer red path than is already present in the graph, i.e., any path going through an edge $(p,q) \in E(A_i)$ can be increased in length by selecting a vertex from $r \in E(B_i)$ using the edges $(p,r)$ and $(r,q)$ instead. This is always possible since $|E(B_i)|$ is sufficiently larger than $|E(A_i)|$.
\qed 
\vskip 0.2cm

Finally, we show how constructing (red-)transmitters where $x = k - 1$ is sufficient to show the existence of all defined transmitters.

\begin{figure}[t]
    \centering
    \includegraphics[width=0.8\textwidth]{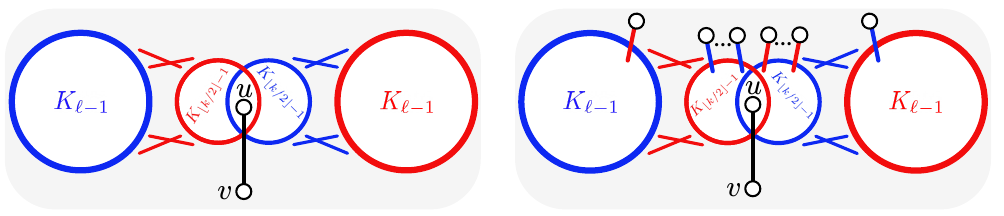}
    \caption{$(k,k-1)$-transmitters for even $k$ (left) and odd $k$ (right).}
    \label{fig:kk-transmit}
\end{figure}

\begin{corollary}
For valid $k, \ell$, and $x$,
$(k,\ell,x)$-blue-transmitters and $(k,\ell,x)$-red-transmitters exist.
\end{corollary}
\begin{proof}
    Let $H$ be a $(k,\ell,k-1)$-red-transmitter where $u \in V(H)$ is the strict endpoint of a red $P_{k-1}$ in all of $H$'s good colorings.
    For valid $x$, 
    the $(H, u, x-1)$-thread graph $G$ is a $(k,\ell,x)$-blue-transmitter, where the thread-end $v$ is the strict endpoint of a blue $P_{x}$ in all good colorings of $G$; to avoid constructing red $P_k$'s each edge along the path of $u_i$'s is forced to be blue by the red $P_{k-1}$ from $H$, where $u_i$ is the vertex $u$ in the $i^{th}$ copy of $H$ as defined in Definition~\ref{def:thread}.
    
    To construct a $(k,\ell,x)$-red-transmitter, we use a similar construction.
    Let $H$ be a $(k,\ell,\ell-1)$-blue-transmitter where $u \in V(H)$ is the strict endpoint of a blue $P_{\ell-1}$ in all good colorings of $H$.
    For valid $x$, 
    the $(H, u, x)$-thread graph $G$ is a $(k,\ell,x-1)$-red-transmitter, where the thread-end $v$ is the strict endpoint of a red $P_{x}$ in all good colorings of $G$. \qed
\end{proof}

\begin{corollary}
For valid $k$ and $x$,
$(k,x)$-transmitters exist.
\end{corollary}
\begin{proof}
    Let $H$ be a $(k,k-1)$-transmitter where $u \in V(H)$ is the strict endpoint of a red/blue $P_{k-1}$ in all of $H$'s good colorings.
    For valid $x$, 
    the $(k, u, x-1)$-thread graph $G$ is a $(k,x)$-transmitter, where the thread-end $v$ is the strict endpoint of a red or blue $P_{x}$ in all good colorings of $G$. 
    Let $u_i$ be the vertex as defined in Definition~\ref{def:thread}.
    Each $u_i$ is the strict endpoint of $P_{k-1}$ of the same color; otherwise, the edge between two $u$'s cannot be colored without forming a red or blue $P_{k}$. Thus, 
    each such edge must be colored red (resp., blue) by the blue (resp., red) $P_{k-1}$ coming from $H$. \qed
\end{proof}

\section{Conclusion and Future Work}
\label{sec:conclude}
A major and very difficult goal is to classify the complexity for $(F,H)$-Arrowing for all fixed $F$ and $H$. We conjecture that in this much more general case a dichotomy theorem still holds, with these problems being either in P or coNP-complete. This seems exceptionally difficult to prove. 
To our knowledge, all known dichotomy theorems 
for graphs classify the problem according to one fixed graph, and the polynomial-time characterizations are much simpler than in our case. We see this paper as an important first step in accomplishing this goal.

\section*{Acknowledgments}

This work was supported in part by grant NSF-DUE-1819546.
We would like to thank the anonymous reviewers for their valuable comments.
\bibliographystyle{splncs04}
\bibliography{mybib}

\begin{thebibliography}{10}
\providecommand{\url}[1]{\texttt{#1}}
\providecommand{\urlprefix}{URL }
\providecommand{\doi}[1]{https://doi.org/#1}

\bibitem{ach:j:colorability}
Achlioptas, D.: {The Complexity of {$G$}-Free Colorability}. Discrete
  Mathematics  \textbf{165-166},  21--30 (1997)

\bibitem{antunes2019characterizing}
Antunes~Filho, I.T.F.: Characterizing Boolean Satisfiability Variants. Ph.D.
  thesis, Massachusetts Institute of Technology (2019)

\bibitem{berman200322sat}
Berman, P., Karpinski, M., Scott, A.: {Approximation Hardness of Short
  Symmetric Instances of MAX-3SAT}. ECCC  (2003)

\bibitem{bikov2018}
Bikov, A.: Computation and Bounding of {F}olkman Numbers. Ph.D. thesis, Sofia
  University ``St. Kliment Ohridski'' (06 2018)

\bibitem{Bu3}
Burr, S.A.: {On the Computational Complexity of Ramsey-Type Problems}.
  Mathematics of Ramsey Theory, Algorithms and Combinatorics  \textbf{5},
  46--52 (1990)

\bibitem{burr1976graphs}
Burr, S.A., Erd\H{o}s, P., Lov{\'a}sz, L.: {On graphs of {R}amsey type}. Ars
  Combinatoria  \textbf{1}(1),  167--190 (1976)

\bibitem{burr1985use}
Burr, S.A., Ne{\v{s}}et{\v{r}}il, J., R{\"o}dl, V.: On the use of {S}enders in
  {G}eneralized {R}amsey {T}heory for {G}raphs. Discrete Mathematics
  \textbf{54}(1),  1--13 (1985)

\bibitem{for-hop-wyl:j:subgraph-homeomorphism}
Fortune, S., Hopcroft, J., Wyllie, J.: {The Directed Subgraph Homeomorphism
  Problem}. Theoretical Computer Science  \textbf{10},  111--121 (1980)

\bibitem{gerencser1967ramsey}
Gerencs{\'e}r, L., Gy{\'a}rf{\'a}s, A.: On {R}amsey-type problems. Ann. Univ.
  Sci. Budapest. E{\"o}tv{\"o}s Sect. Math  \textbf{10},  167--170 (1967)

\bibitem{hel-nes:j:H-coloring}
Hell, P., Ne{\u{s}}et{\u{r}}il, J.: {On the Complexity of $H$-Coloring}.
  Journal of Combinatorial Theory, Series B  \textbf{48},  92--110 (1990)

\bibitem{Ho}
Hook, J.: {Critical graphs for \emph{R(P\({}_{\mbox{n}}\),
  P\({}_{\mbox{m}}\))}} and the star-critical {R}amsey number for paths.
  Discussiones Mathematicae Graph Theory  \textbf{35}(4),  689--701 (2015)

\bibitem{krom1967decision}
Krom, M.R.: {The Decision Problem for a Class of First-Order Formulas in which
  all Disjunctions are Binary}. Mathematical Logic Quarterly  \textbf{13}(1-2),
   15--20 (1967)

\bibitem{le_et_al:LIPIcs.MFCS.2022.68}
Le, H.O., Le, V.B.: {Complexity of the Cluster Vertex Deletion Problem on
  H-Free Graphs}. In: MFCS 2022. vol.~241, pp. 68:1--68:10 (2022)

\bibitem{mey-sto-ph}
Meyer, A.R., Stockmeyer, L.J.: The equivalence problem for regular expressions
  with squaring requires exponential space. In: IEEE SWAT. pp. 125--129 (1972)

\bibitem{ds1}
Radziszowski, S.: {Small Ramsey Numbers}. Electronic Journal of Combinatorics
  \textbf{DS1},  1--116 (January 2021), \url{https://www.combinatorics.org/}

\bibitem{rosta}
Rosta, V.: {Ramsey Theory Applications}. Electronic Journal of Combinatorics
  \textbf{DS13},  1--43 (December 2004), \url{https://www.combinatorics.org/}

\bibitem{rut:c:graph-coloring}
Rutenburg, V.: {Complexity of Generalized Graph Coloring}. In: {MFCS 1986}.
  Lecture Notes in Computer Science, vol.~233, pp. 573--581. Springer (1986)

\bibitem{Scha}
Schaefer, M.: {Graph Ramsey Theory and the Polynomial Hierarchy}. Journal of
  Computer and System Sciences  \textbf{62},  290--322 (2001)

\end{thebibliography}

\section*{Appendix}

\subsection*{Clause Gadget for $(P_3, P_\ell)$- and $(P_4, P_4)$-Nonarrowing}

In Figure~\ref{fig:p3pl-clause}, we present the clause gadget for $(P_3, P_\ell)$-Nonarrowing.
In Figure~\ref{fig:p4p4}, we present the clause gadget for $(P_4, P_4)$-Nonarrowing, inspired by the gadget used in Rutenburg's hardness proof~\cite{rut:c:graph-coloring} for the same problem.

\begin{figure}[ht]
    \centering
    \includegraphics[width=0.8\textwidth]{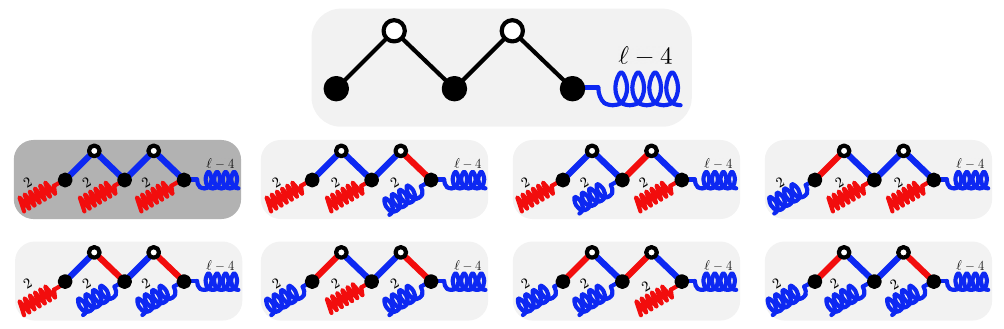}
    \caption{The clause gadget for $(P_3, P_\ell)$-Nonarrowing. The format is similar to Figure~\ref{fig:pkpl-clause}.}
    \label{fig:p3pl-clause}
\end{figure}

\begin{figure}[ht]
    \centering
    \includegraphics[width=0.8\textwidth]{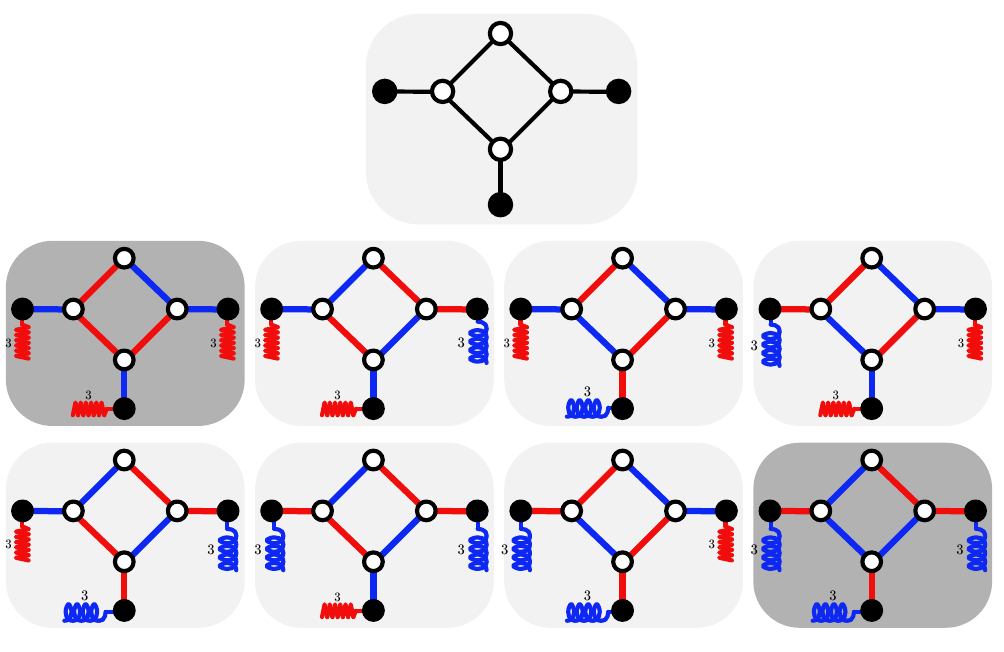}
    \caption{The clause gadget for $(P_4, P_4)$-Nonarrowing. This is similar to the gadget used in Rutenburg's hardness proof~\cite{rut:c:graph-coloring}, wherein the clause gadget was just the $C_4$.}
    \label{fig:p4p4}
\end{figure}

\subsection*{Existence of Transmitters}

We first state Hook's theorem for the symmetric case (where $k = \ell$) and state an observation about the properties of said colorings.

\begin{theorem}[Hook~\cite{Ho}]
\label{thm:CritK-eq}
Let $k \geq 4$ and $r = R(P_k, P_k) - 1$.
The possible $(P_k, P_k)$-good colorings of $\RG$ can be categorized into two types. In each case, $V(G)$ is partitioned into sets $A$ and $B$. The types are defined as follows:
\begin{itemize}
    \item Type 1a. 
    Let $|A| = \lfloor k/2 \rfloor - 1$ and $|B| = k - 1$.
    Each edge in $E(B)$ must be blue, and each edge in $E(A, B)$ must be red. Any coloring of $E(A)$ is allowed.
    \item Type 1b. Like Type 1a, but $E(B)$ is red, and $E(A, B)$ is blue.
    \item Type 2a. 
    Let $|A| = \lfloor k/2 \rfloor - 1$ and $|B| = k - 1$, and let $b \in E(B)$. Each edge in $E(B) \setminus \{b\}$ must be blue, and each edge in $E(A, B) \cup \{b\}$ must be red. Any coloring of $E(A)$ is allowed.
    \item Type 2b. Like Type 2a, but $E(B) \setminus {b}$ is red, and $E(A, B) \cup {b}$ is blue.
\end{itemize}

If $k$ is even, then $\RG$ can only have Type 1a/b colorings.
If $k$ is odd, then $\RG$ can have all types of colorings.
\end{theorem}

\begin{observation}
\label{obs:k-k-types}
Let $k \geq 4$ and $r = R(P_k, P_k) - 1$. 

\begin{itemize}
    \item In Type 1a/1b $(P_k, P_k)$-good colorings of $K_r$: (1) each vertex in $B$ is a strict endpoint of a blue/red $P_{k - 1}$, (2) when $k$ is even (resp., odd), each vertex in $B$ is a strict endpoint of a red/blue $P_{k - 1}$ (resp., $P_{k-2}$), and (3) when $k$ is even (resp., odd), each vertex in $A$ is a strict endpoint of a red/blue $P_{k - 2}$ (resp., $P_{k-3}$).
    \item In Type 2a/2b $(P_k, P_k)$-good colorings of $K_r$: (1) each vertex in $B$ is a strict endpoint of a blue/red $P_{k - 1}$, (2) each vertex in $B$ is a strict endpoint of a red/blue $P_{k - 1}$, and (3) each vertex in $A$ is a strict endpoint of a red/blue $P_{k - 2}$.
\end{itemize}
\end{observation}

We now present the proofs of Lemmas~\ref{lem:k-l-transmit} and~\ref{lem:k-k-transmit}.
The term $Ai$-vertex (resp., $Bi$-vertex) for $i \in\{1,2,3\}$ is used to refer to vertices belonging to set $A$ (resp., $B$) in a $\RG$ with a Type $i$ coloring, as defined in Theorem~\ref{thm:CritK-uneq}. \\

\noindent\textit{Proof of Lemma~\ref{lem:k-l-transmit}.}
We first consider the case where $k = 3$.
Let $G$ be the graph constructed by attaching a leaf vertex, $v$, to a $K_{\ell-1}$. It is easy to see that every vertex in a $(P_3, P_\ell)$-good coloring of $K_{\ell-1}$ is the strict endpoint of blue $P_{\ell-1}$. Thus, the leaf edge must be red, and $v$ is the strict endpoint of a red $P_2$ is all good coloring of $G$.
Now, suppose $k \geq 4$.
Let $r = R(P_k, P_\ell) - 1$. We consider three cases. \\

\noindent \textit{Case 1 ($k$ is even).} Covered in the main text. \\

\noindent \textit{Case 2 ($k$ is odd and $\ell > k + 1$).}
Type 1 and 2 colorings of $\RG$ are allowed in this case.
Let $H$ be the graph constructed by attaching a leaf node to $\lfloor k/2 \rfloor$ vertices of a $\RG$.
We refer to these leaf vertices as $L$-vertices. We now 
analyze the properties of $H$. 
First, note that at least one vertex in $B$ must be an $L$-vertex, since $|A| = \lfloor k/2 \rfloor - 1$. Recall that in both Type 1 and 2 colorings, a $Bi$-vertex is a strict endpoint of blue $P_{\ell-1}$.  
Thus, the leaf edge must be red to avoid making a blue $P_{\ell}$. This implies that the $\RG$ in $H$ must have a Type 1 coloring; otherwise, the red leaf edge and red $P_{k-1}$ from the Type 2 coloring would form a red $P_{k}$. Note that 
if at least two $L$-vertices were adjacent to vertices in $B$, then two red leaf edges and the red $P_{k-2}$ would form a red $P_k$. Thus, there is exactly one $L$-vertex adjacent to a vertex in $B$.
The red leaf edge adjacent to a $B$-vertex, along with the red $P_{k-3}$ (formed using edges in $E(A,B)$), makes each vertex in $A$ the strict endpoint of a red $P_{k-2}$.
$H$ can now emulate the role of $K_r$ with Type 1 colorings as in the previous case. 

Let $H'$ be the graph constructed by attaching a leaf node to $\lfloor k/2 \rfloor - 1$ vertices of an $\RG$.
Consider the graph $F$, where two disjoint $H'$'s are connected by a single edge $(u,v)$, where $u$ and $v$ are members of a $K_r$ and not adjacent to $L$-vertices.
As in the previous case, it is easy to see that
$u$ and $v$ must be $A1$-vertices, and $(u,v)$ must be blue.

We define the $(k,\ell,k-1)$-red-transmitter,
$G$, as the $(H', u, \ell-1)$-thread graph, where $u$ is a
member of a $K_r$ not adjacent to any $L$-vertex.
The argument from the previous case shows that the thread-end of $G$ is the strict endpoint of a red $P_{k-1}$.
Moreover, the coloring where each $E(A)$ is colored red and each edge 
between an $L$-vertex and an $A1$-vertex
is colored blue is a $(P_k, P_\ell)$-good coloring. \\

\noindent \textit{Case 3 ($k$ is odd and $\ell = k + 1$).}
Note that in this case, all three types of colorings of $\RG$ are allowed.
We consider two subcases.

\noindent\textit{Case 3.i ($k \geq 7$).}
Consider the graph $H$ where two $\RG$'s share a single vertex, $p$, and $q$ is a leaf vertex connected to $p$.
Let $X$ and $Y$ refer to the vertex set of each $\RG$.
We first show that $p$ cannot be a $Bi$-vertex for any $i \in \{1,2,3\}$ in both $X$ and $Y$. We prove via contradiction: assume without loss of generality that $p$ is a $Bi$-vertex in $X$.
Recall that a $Bi$-vertex is the endpoint of a blue $P_{\ell-1}$ and red $P_{k-2}$. Thus, $p$ must be an $Aj$-vertex in $Y$, for some $j\in \{1,2,3\}$. However, $A1$- and $A2$-vertices are endpoints of red $P_{k-3}$'s, and $A3$-vertices are endpoints of blue $P_{\ell-2}$'s. If $p$ is an $A1$- or $A2$-vertex, then a red $P_{k-3+k-2-1}$ is formed, which is forbidden when $k \geq 6$, or a blue $P_{\ell-2 + \ell - 1 - 1}$ is formed, which is forbidden when $\ell \geq 4$. Thus, $p$ cannot be a $Bi$-vertex in $X$ or $Y$.
Also note that if $p$ is an $A1$- or $A2$-vertex in both $X$ and $Y$, then a red $P_{k-3+k-3-1}$ is formed, which is forbidden when $k \geq 7$. Similarly, $p$ cannot be an $A3$-vertex in both $X$ and $Y$ otherwise a blue $P_{\ell - 2 + \ell - 2 - 1}$ is formed. 
Thus, $p$ must be an $A1$- or $A2$-vertex in $X$, and an $A3$-vertex in $Y$, or vice versa. This implies that $p$ must be the strict endpoint of a red $P_{k-3}$ (or $P_{k-2}$) and a blue $P_{\ell-2}$.

Let $G$ be the graph constructed as follows. Take two copies of $H$ and contract the vertices labeled $q$. Then, attach a leaf vertex, $v$, to the contracted vertex.
Let $p_1$ and $p_2$ be the vertices in the intersection of two $\RG$'s. Observe that $(p_1, q)$ and $(p_2, q)$ must be different colors, otherwise a red $P_{k-3+k-3 + 1}$ or a blue $P_{\ell-2 + \ell-2 + 1}$ is formed when both are red or blue, respectively. Assume without loss of generality that $(p_1, q)$ is red and $(p_2,q)$ is blue. Since $p_1$ and $p_2$ are the strict endpoints of blue $P_{\ell-2}$'s, $q$ must be the strict endpoint of a blue $P_{\ell-1}$, forcing $(q,v)$ to be red. Observe that $(p_1,q)$ is the strict endpoint of a red $P_{k-2}$ or $P_{k-1}$ depending on whether $p_1$ is an $A1$ or $A2$-vertex in one of the $\RG$'s. Clearly, $p_1$ must be an $A1$-vertex, otherwise a red $P_k$ is formed with $(q,v)$. Therefore, $v$ must be the strict endpoint of a red $P_{k-1}$.
Finally, we describe a good coloring of $G$: color $E(A)$ red if the $K_r$ has a Type 1 coloring, and blue if it has a Type 3 coloring.

\noindent\textit{Case 3.ii ($k < 7$).} 
Note that $k = 5$ and $\ell = 6$ is the only possibility in this case. 
The transmitter for this is shown in Figure~\ref{fig:kl-transmit}. 
Let $X, Y,$ and $Z$ be the vertex sets of the $K_r$'s as in Figure~\ref{fig:kl-transmit}.
First, note that $Z$ cannot have a Type 1 or Type 2 coloring. We prove via contradiction. 
Assume that $Z$ has a Type 1 or 2 coloring.
Since $k = 5$, $|A| = \lfloor 5/2 \rfloor = 1$ for Type 1 and 2 colorings, at least one vertex in $\{a,t\}$ must be a $B1$- or $B2$-vertex in $Z$. Assume without loss of generality that $a$ is a $B1$ or $B2$-vertex. Note that $a$ is the strict endpoint of a blue $P_{\ell-1}$ from $Z$ (Observation~\ref{obs:k-l-types}). Thus, each edge connected to $a$ in $X$ must be red. 
Since $a$ is not adjacent to any blue edge, it can only be an $A1$ or $A2$-vertex in $X$. Since $|A| = 1$ in Type 1 and 2 colorings, $b$ must be a $B1$ or $B2$-vertex and be the strict endpoint of a blue $P_{\ell-1}$, implying that $(b, c)$ is red. Since $a$ is also a strict endpoint of a red $P_{k-2}$ from $Z$, the edges $(a, b)$ and $(b, c)$ would form a red $P_{k}$. 
Thus, $Z$ must have a Type 3 coloring. Note that $a$ and $t$ must be $A3$-vertices in $Z$, otherwise the edges $(a, b)$ and $(t, u)$ would form red $P_k$'s or blue $P_\ell$'s.
$a$ (resp., $t$) must be an $A1$- or $A2$-vertex in $X$ (resp., $Y$) because every other type of vertex is the strict endpoint of a blue $P_{\ell-2}$ in $X$ (resp., $Y$). This would form a blue $P_{\ell -2 + \ell -2 }$ with the $P_{\ell-2}$ from $Z$.
Since $|A| = 1$ in Type 1 and 2 colorings, $b$ and $u$ must be in $B1$- or $B2$-vertices. $X$ and $Y$ must have Type 1 colorings otherwise $(b,c)$ and $(u, v)$ make red $P_k$'s, on account of $B2$-vertices being strict endpoints of red $P_{k-1}$'s and the leaf edges being red.
Thus, $b$ and $u$ are $A1$-vertices and are strict endpoints of red $P_{k-2}$'s, making $c$ and $v$ strict endpoints of red $P_{k-1}$'s. Finally, $(a,t)$ must be colored blue to obtain a good coloring of $G$.
\qed
\vskip 0.2cm

We now prove Lemma~\ref{lem:k-k-transmit}. Since the constructions are similar to that of the previous lemma, some details are skipped since the same arguments can be applied \textit{mutatis mutandis}. 

\vskip 0.2cm
\noindent\textit{Proof of Lemma~\ref{lem:k-k-transmit}.} 
When $k = 3$, $K_2$ is trivially a transmitter. Suppose $k \geq 4$. Let $r = R(P_k, P_k) - 1$. We consider two cases.

\noindent\textit{Case 1 ($k$ is even).} Recall that only Type 1a and 1b colorings are allowed in this case.
As in Case 3.i of Lemma~\ref{lem:k-l-transmit}'s proof, 
consider the graph $G$ where two $\RG$'s share a single vertex, $u$, and $v$ is a leaf vertex connected to $p$.
Let $X$ and $Y$ refer to the vertex set of each $\RG$.
Note that $p$ cannot be a $B1$-vertex in either $K_r$; 
a $B1$-vertex is the strict endpoint of a red and blue $P_{k-1}$ (Observation~\ref{obs:k-k-types}), and this would form a red/blue $P_k$ with $(p,q)$. 
Thus, $q$ is a Type 1 vertex in $X$ and $Y$.
Also, note that $q$ must be a Type 1a vertex in $X$ and a Type 1b vertex in $Y$ (or vice versa).
Otherwise, if both $X$ and $Y$ are of the same type, then a red/blue $P_{k-2 + k-2}$ is formed, which is forbidden when $k \geq 4$.
Thereby, $q$ is the strict endpoint of a red and blue $P_{k-2}$ in all good colorings of $G$. 
Moreover, $v$ is the strict endpoint of a red/blue $P_{k-1}$ depending on the color of $(p,v)$.
Thus, $G$ is a $(k, k-1)$-transmitter where $v$ is the strict endpoint of a red/blue $P_{k-1}$ in all good colorings of $G$.
Finally, note that $G$ is $(P_k, P_k)$-good since all edges in $E(A)$ can be colored red (resp., blue) in the $K_r$ with the Type 1a (resp., Type 1b) coloring. \\

\noindent\textit{Case 2 ($k$ is odd).}
Let 
$H$ be the graph constructed by attaching a leaf node to $\lfloor k/2 \rfloor - 1$ vertices of a $\RG$.
Consider the graph $G$, constructed by taking two copies of $H$, contracting two vertices not adjacent to $L$-vertices, and attaching a leaf vertex, $v$, to the contracted vertex, denoted as $u$. 
Let $X$ and $Y$ refer to the vertex set of each $\RG$.
As argued in Case 2 of Lemma~\ref{lem:k-l-transmit}'s proof, both $K_r$'s in $G$ must have Type 1 colorings. 
Moreover, $v$ is the strict endpoint red/blue $P_{k-2}$ from both $K_r$'s.
As argued in the even case, $X$ and $Y$ must have different types of colorings to avoid a red/blue $P_k$ going across both $H$'s.
Thus, $G$ is $(k,k-1)$-transmitter, and $v$ is the strict endpoint of a red/blue $P_{k-1}$ in all good colorings of $G$.
Finally, note that $G$ is $(P_k, P_k)$-good using the coloring from the previous case.
\qed 

\subsection*{Proofs for Observations~\ref{obs:k-l-types} and~\ref{obs:k-k-types}}

The observations are simple corollaries of the following lemmas.

\begin{lemma}
\label{lem:p2n1}
Let $G$ be a graph whose vertices are partitioned into $M$ and $N$ such that (1) $|M| = k$ and $|N| \geq k+1$, (2) $E(M,N)$ includes all possible edges, and (3) $N$ is an independent set.
Then, each vertex in $N$ is an endpoint of a 
$P_{2k+1}$.
Moreover, this is the largest path in $G$ that each $n \in N$ is the endpoint of.
\end{lemma}
\begin{proof}
Let $m_i \in M$ and $n_i \in N$. We first show that 
each $n_i \in N$ is the endpoint of $P_{2k+1}$. Consider the path
$n_1, m_1, n_2, \ldots, m_k, n_{k+1}$, which alternates between vertices in $M$ and $N$. Clearly this path is of size $2k+1$, and must exist because $|M| \geq k+1$ and all edges in $E(M,N)$ exist. 
It is easy to see that the vertices may be relabeled so that any vertex in $N$ can be the endpoint of this $P_{2k+1}$.
We now show that no path larger than $2k+1$ exists in $G$, via contradiction. Assume there exists a path $Q$ on $2k+2$ vertices in $G$. Since $|M| = k$, at least $2k+2 - |M| = k+2$ vertices of $Q$ must be in $N$.
Recall that in any $P_\ell$ there are two vertices of degree one and $\ell-2$ vertices of degree two, and there are a total of $\ell-1$ edges.
If at least $k+2$ vertices in $N$ are in $Q$, then $1 + 1 + 2k$ edges of $Q$ must be in $E(M,N)$, since $N$ is an independent set. However, this implies that $Q = P_{2k+2}$ has at least $2k+2$ edges, which is a contradiction.
\end{proof}

\begin{lemma}
\label{lem:p2n}
Let $G$ be a graph who vertices are partitioned into $M$ and $N$ such that (1) $|M| = k$ and $|N| \geq k$, (2) $E(M,N)$ includes all possible edges, and (3) $N$ is an independent set.
Then, each vertex in $M$ is an endpoint of a
$P_{2k}$.
Moreover, this is the largest path in $G$ that each $m \in M$ is the endpoint of.
\end{lemma}
\begin{proof}
Let $m_i \in M$ and $n_i \in N$. We first show that 
each $m_i \in M$ is the endpoint of $P_{2k}$. Consider the path
$n_1, m_1, n_2, m_2, \ldots, n_k, m_k$, which alternates between vertices in $M$ and $N$. Clearly this path is of length $2k$, and must exist because $|M| = k$ and all edges in $E(M,N)$ exist. It is easy to see that the vertices may be relabeled so that any vertex in $M$ can be the endpoint of this $P_{2k}$.
We now show that no path larger than $2k$ with an endpoint in $M$ exists in $G$, via contradiction. 
Assume there exists a path $Q$ on $2k+1$ vertices in $G$, with one endpoint in $M$.
Since $|M| = k$, at least $2k+1 - |M| = k+1$ vertices of $Q$ must be in $N$.
Recall that in any $P_\ell$ there are two vertices of degree one and $\ell-2$ vertices of degree two, and there are a total of $\ell-1$ edges.
If at least $k+1$ vertices in $N$ are in $Q$ and at most one of them may be an endpoint of $Q$, then $1 + 2k$ edges of $Q$ must be in $E(M,N)$, since $N$ is an independent set. However, this implies that $Q = P_{2k+1}$ has at least $2k+1$ edges, which is a contradiction.
\end{proof}

\begin{lemma}
\label{lem:KkAllVPk}
For all $k \geq 1$, each vertex in $K_k$ is an endpoint of a $P_k$.
\end{lemma}
\begin{proof}
The statement is trivially true for $k = 1$. Assume it is true for all $k < n$. For $k = n$, let $v$ be any vertex in $V(K_n)$. By the inductive hypothesis, there must be a $P_{n-1}$ in $V(K_n) \setminus \{v\}$. $v$ must be connected to an endpoint of said $P_{n-1}$, implying that $v$ is an endpoint of a $P_n$.
\end{proof}

\begin{lemma}
\label{lem:Kk-eAllVPk}
For all $k \geq 4$, every vertex in $K_k-e$ is an endpoint of a $P_k$. 
\end{lemma}
\begin{proof}
Let $v$ and $w$ be the only two vertices in $K_k - e$ that do not share an edge. 
Let $u$ be any vertex in $V(K_k - e) \setminus \{v,w\}$. 
Let $V' = V(K_k - e) \setminus \{u,v,w\}$. 
Since $k \geq 4$, $|V'| \geq 1$ and $G[V']$ is the complete graph $K_{k-3}$. Clearly, $K_{k-3}$ must have a $P_{k-3}$. 
Observe that $w$, the $P_{k-3}$ in $G[V']$, $u$, and $v$ form a $P_k$ with endpoints $w$ and $v$.
Moreover, $u$, $v$, the $P_{k-3}$, and $w$ form a $P_k$ with endpoints $u$ and $w$.
\end{proof}

In Type 1/1a/1b and Type 3 colorings, the observations made are simple applications of the lemmas above.
For Type 2/2a/2b colorings, it is easy to see that the extra red edge in $E(B)$ increases the length of the red paths in $K_r$.
\end{document}